\documentclass[conference,twocolumn,final]{IEEEtran}
\makeatletter
\def\ps@headings{%
\def\@oddhead{\mbox{}\scriptsize\rightmark \hfil \thepage}%
\def\@evenhead{\scriptsize\thepage \hfil \leftmark\mbox{}}%
\def\@oddfoot{}%
\def\@evenfoot{}}
\makeatother 

\usepackage[final]{graphicx}
\usepackage[reqno]{amsmath}
\usepackage{amssymb}
\usepackage{cite}
\usepackage{balance}
\usepackage{color}
\usepackage{caption}

\usepackage[german,vlined,boxed]{algorithm2e}
\newenvironment{algorithmic}{%
\algorithm
}{%
\endalgorithm
}
\usepackage{floatrow}

\newfont{\bbb}{msbm10 scaled 500}

\newfont{\bb}{msbm10 scaled 1100}

\newcommand{\Fc}{{\cal F}}

\newcommand{\argmax}{\operatornamewithlimits{argmax}}

\newtheorem{theorem}{Theorem}
\newtheorem{proposition}[theorem]{Proposition}
\newtheorem{definition}[theorem]{Definition}

\IEEEoverridecommandlockouts

\author
{Mehmet Karaca, Tansu Alpcan, Ozgur Ercetin}
\title{Smart Scheduling and Feedback Allocation over Non-stationary Wireless Channels
\thanks{
Mehmet Karaca and, Ozgur Ercetin are with the Faculty of Engineering
and Natural Sciences, Sabanci University, 34956 Orhanli-Tuzla,
Istanbul/Turkey. Email: mehmetkrc@su.sabanciuniv.edu,
oercetin@sabanciuniv.edu.}
\thanks{
Tansu Alpcan is with the Dept. of Electrical and Electronic
Engineering, The University of Melbourne, Australia. Email:
tansualpcan@gmail.com }
\thanks{This work is supported in part by European Commission IRSES program under grant AGILENet.}}

\begin{document}
\maketitle


\begin{abstract}
It is well known that opportunistic scheduling algorithms are
throughput optimal under dynamic channel and network conditions.
However, these algorithms achieve a hypothetical rate region which
does not take into account the overhead associated with channel
probing and feedback required to obtain the full channel state
information at every slot. In this work, we design a joint
scheduling and channel probing algorithm by considering the overhead
of obtaining the channel state information. We adopt a correlated
and non-stationary channel model, which is more realistic than those
used in the literature. We use concepts from learning and
information theory to accurately track channel variations to
minimize the number of channels probed at every slot, while
scheduling users to maximize the achievable rate region of the
network. Simulation results show that with the proposed algorithm,
the network can carry higher user traffic.

\end{abstract}


\section{Introduction}
\label{sec:intro} In wireless networks, the channel conditions are
time-varying due to the fading and shadowing. Opportunistic
scheduling algorithms take advantage of favorable channel conditions
in assigning time slots to users. Optimal scheduling in wireless
networks has been extensively studied in the literature under
various assumptions. The seminal work by Tassiulas and Ephremides
have shown that a simple opportunistic algorithm that schedules the
user with the highest queue backlog and transmission rate product at
every time slot, can stabilize the network, whenever this is
possible~\cite{MW}.

A common assumption in the literature on opportunistic algorithms is
that the {\em exact} and {\em complete} channel state information,
(CSI) of all users is available at every time slot. Hence, these
algorithms achieve a \textit{hypothetical} rate region by assuming
that full channel state information is available without any channel
probing or feedback costs.  However, in practice acquiring CSI
introduces significant overhead to the network, since CSI is
obtained either by probing the channel or via feedback from the
users. In current wireless communication standards such as WiMax and
LTE there is a feedback channel used to relay CSI from the users to
base station. Obviously, this feedback channel is bandlimited and it
is impossible to obtain CSI from all users at the same slot.

Another common assumption that does not hold in practice is the
wireless channel being independent and identically distributed
(iid), and being governed by a \textit{stationary} stochastic
process. The most common assumption is that the channel can be
modeled by a stationary Markov chain. The measurement study
in~\cite{Bozidar:nonstationary11} shows that the wireless channel
exhibits time-correlated and non-stationary behavior.

In this work, we develop a joint scheduling and channel probing
algorithm for time-correlated and non-stationary channels. The
channel probing is based on Gaussian Process Regression (GPR)
technique~\cite{Rasmussen:GP}, which is used to learn and track the
wireless channel. The scheduling part is based on well known
Max-Weight algorithm. The joint algorithm dynamically determines the
set of channels that must be probed at every time slot based on the
information obtained from the previous channel observations, and
then schedules a node based on the obtained CSI and queue states. We
show that GPR-based probing works well for realistic,
time-correlated and non-stationary wireless channels at
significantly lower probing cost.

Our contributions are summarized as follows:
\begin{itemize}
\item We use information theoretical concepts to quantify the
uncertainty in the channel state under finite and infrequent
measurements.
\item Based on the work in~\cite{Alpcan:Valuetools11}, Gaussian Process Regression learning algorithm is proposed to track the channel evolution.
\item A joint scheduling and probing algorithm is proposed in which the subset of users
probed at every slot is adaptively selected based on the dynamics of
the channel processes.
\item We implement a realistic network setting where we simulate High Data Rate (HDR) protocol in CDMA cellular
networks, and wireless channel is modeled as time-correlated and
non-stationary.  We show by numerical analysis that when our
proposed algorithm is used the network can carry higher user traffic
compared to Max-Weight algorithm with full CSI.
\end{itemize}

\section{Related Works}
\label{sec:related} 
In~\cite{Gopalan:allerton07}, the authors propose a
throughput-optimal algorithm when channel distribution is known.
In~\cite{Ouyang:mobihoc11}, the authors present a joint algorithm
for multi-channel system with limited feedback bandwidth. The joint
scheduling and probing  problem is transformed to multi armed bandit
problem in~\cite{Quyang:infocom11}. In~\cite{Chaporkar:mobihoc09},
channel probing is performed at the beginning of transmission by
taking a portion of time slot. Then, the problem of finding optimal
joint algorithm is transformed into an optimal stopping time problem
and is solved by Markov Decision Process (MDP).  Aforementioned
works assume that the underlaying stochastic process of the channel
evolves according to a fixed stationary process such as ergodic
Markov chain. In practice, such an assumption does not hold most of
the time. In addition, the authors in~\cite{Hallen:fading} proposed
a technique to estimate future values of the fading coefficient of a
non-stationary channel. The proposed technique is based on the
autoregressive (AR) model with order $p$. According to AR model, the
current CSI of a user can only be determined when $p$ previous CSIs
of that user are given. Regarding to a joint scheduling and probing
problem, the corresponding user must be probed at every previous $p$
slots to obtain $p$ previous CSIs. However, a joint algorithm does
not necessarily probe a user at every time slot. Therefore, the
proposed technique is not suitable for a scheduling problem with
limited feedback.

There are very few studies which propose a scheduling algorithm for
non-stationary channels. In~\cite{Andrews:Nonstationary}, the
authors showed that with Max-Weight algorithm the average queue
sizes increases exponentially with the number of users. It was
assumed that channel state information of each user is available at
the scheduler at every time slot.

\section{System Model and Problem Formulation}
\label{sec:model} We consider a cellular system with a single base
station transmitting to $N$ users. Let $\mathcal{N}$ denote the set
of users in the cell. Time is slotted, $t\in \{0,1,2,\ldots\}$, and
wireless channel between the base station and a mobile user is
modeled as a time-correlated fading process. The gain of the channel
is constant over the duration of a time slot but varies between
slots. Let $C_n(t)$ denote CSI of user $n$ at time slot $t$.
$C_n(t)$ is a random process which may or may not have a stationary
probability distribution. Let $c_n(t)$ represent the realization of
$C_n(t)$ at time $t$, $n\in \{1,2,\dots,N\}$. In the rest of the
paper, we use channel and user interchangeably.

Let $\mu_n(c_n(t))$, or simply $\mu_n(t)$, denote the transmission
rate of user $n$ which depends on CSI of that user, and is bounded
as $\mu_{min}<\mu_n(t)<\mu_{max}$. We assume that at each time slot
at most one user can be scheduled to receive data from the base
station. The base station transmits to users at  fixed power, so
transmission rate of each user only depends on $c_n(t)$.

Let $a_n(t)$ be the amount of data (bits or packets) arriving into
the queue of user $n$ at time slot $t$. We assume that $a_n(t)$ is a
stationary process and it is independent across users and time
slots. We denote the arrival rate vector as $\boldsymbol
\lambda=(\lambda_1,\lambda_2,\cdots,\lambda_N)$, where $\lambda_n =
{\mathbb E}[a_n(t)]$. Let $\boldsymbol
q(t)=(q_1(t),q_2(t),\cdots,q_N(t))$ denote the vector of queue
sizes, where $q_n(t)$ is  the queue length of user $n$ at time slot
$t$.
\begin{definition}
A queue is strongly stable if
\begin{equation}
\limsup_{t\rightarrow \infty}\frac{1}{t}
\sum_{\tau=0}^{t-1}\mathbb{E}(q_n(t)) < \infty \label{eq:defination}
\end{equation}
\end{definition}
Moreover, if every queue in the network is stable then the network
is called stable.

The operation of the system is as follows.  At the beginning of a
time slot, CSI of a subset of users is obtained by the base station.
Then, the base station schedules a single user out of this subset
for transmission. Here, the only overhead we take into account is
the channel bandwidth and the time used for obtaining CSI. We
consider {\em dynamic} feedback channel allocation as follows:

\textbf{Dynamic feedback channel allocation model:} According to
this model, there is no dedicated feedback channel, and CSI is
relayed over the data channel. Hence, depending on the needs of the
algorithm CSI from varying number of users can be obtained. We
quantify the overhead of obtaining the CSI of a single user in terms
of a time fraction of the time slot. This time duration may include
the time spent for pilot signal transmission, measurement of the
signal strength of pilot signal and the transmission of CSI to the
base station. Assume that $\beta$ fraction of the time slot is
consumed to obtain CSI from a single user. Hence, only
$(1-m\beta)\times T_s$ seconds are available for data transmission
when $m$ users are probed. The amount of data that can be
transmitted by user $n$ is given by,
\begin{align}
d_{n}(t)=(1-m\beta) T_s\times \mu_n(t). \label{eq:d2}
\end{align}

The joint policy $\pi$ selects the triplet $(n,m,S_m)$ under this
model at each time slot $t$, where $n$ is the scheduled user, $S_m$
is the set of probed users and $m$ is the number of users in $S_m$.
We assume that the scheduled user at slot $t$ is selected among the
users probed, i.e., $n\in S_m$. Given $\pi=(n,m,S_m)$, $n$ is
determined according to Max-Weight rule, i.e.,
\begin{align}
n=\argmax_{i \in S_m}\ q_i(t)d_{i}(t).\label{eq:mw-model1}
\end{align}
Let $\Fc$ be the set of feasible policies at a given time slot and
$\pi\in \Fc$.

The amount of data that is transmitted by user $j$ at time slot $t$
under the joint scheduling and probing policy $\pi$,  is given as
follows,
\begin{align}
r_{j}(\pi,t) =& \left\{ \begin{array}{l l}
                    d_{j}(t)              & \text{; if user $j=n$}\\
                    0        & \text{; otherwise}
                \end{array} \label{eq:r2}
    \right.
\end{align}
The dynamics of the queue of user $n$ under scheduling policy $\pi$,
is,
\begin{align}
q_n(t+1)=\max(q_n(t)+a_n(t)-r_{n}(\pi,t))^+ .
\end{align}
where $(x)^+=\max(x,0)$.

\subsection{Problem Formulation} \label{sec:prob_for} We present the
following definitions before discussing the problem formulation.
\subsection{Hypothetical and Functional Rate Regions}
The \textit{achievable rate region} (shortly, rate region) of a
network is defined as the closure of the set of all arrival rate
vectors $\boldsymbol \lambda$ for which there exists an appropriate
scheduling policy that stabilizes the network.
\begin{definition}
$\Lambda_{un}$ is the hypothetical rate region where full CSI is
available (e.g. by an Oracle) without any channel probing or
feedback costs.
\end{definition}

\begin{definition}
$\Lambda_{full}$ is the achievable rate region when probing cost is
taken into account and when all users' channels are probed at every
time slot according to \textit{dynamic} feedback model.
\end{definition}


\begin{definition}
$\Lambda$ is the  rate region under \textit{dynamic} feedback model
when CSI from a subset of users is available.
\end{definition}

\subsection{Optimization Given the Steady-state Channel Distribution}
\label{sec:actual_prob} Our aim is to find a joint scheduling and
channel probing policy that stabilizes the network for a given set
of arrival rates within achievable rate region $\Lambda$ by
dynamically determining a subset of channels probed, and by
scheduling a user from this subset at every time slot. Given the
queue state $\boldsymbol q(t)$, we consider the following
optimization problem:
\begin{align}
 \max_{\pi \in \Fc}&\left\{ {\mathbb E}\left[\sum_{n=1}^N
q_n(t)r_{n}(\pi,t) |\boldsymbol  q(t)\right] \right\},
 \label{eq:problem}
\end{align}

\subsection{Tracking the Instantaneous Channel States}
\label{sec:learn_track} In practice, it is not possible to
accurately determine the exact channel distributions a priori to
system operation. Hence, we propose to use a learning algorithm to
track the channel evolution. Let $\hat{c}_n(t)$ denote the estimated
CSI of user $n$ at the beginning of time $t$. Let $\hat{\mu}_n(t)$
denote the estimated transmission rate of user $n$ at time $t$. One
can replace the actual rates $\mu_n(t)$ by $\hat{\mu}_n(t)$ to
obtain a new set of policies
$\hat{\pi}=(\hat{n},\hat{m},\hat{S}_{\hat{m}})$ according to
\textit{dynamic} feedback model. Also let $\hat{d}_{n}(t) $,
 denote  the amount of data that can be transmitted by user
$n$ by using $\hat{c}_n(t)$ at time $t$. Similarly, the estimated
service rate $\hat{r}_{n}(\hat{\pi},t)$ is defined according to
\eqref{eq:r2} by replacing $d_{n}(t)$ with $\hat{d}_{n}(t)$.

The quality of the estimate of an instantaneous channel state
depends on which users are probed at each slot, i.e.,
$\hat{S}_{\hat{m}}$. Here, we design a joint algorithm that takes
past observations of the channels as an input and determines a
subset of users to be probed at time $t$ so that the channel
estimation error is minimized and the rate region is maximized.

\subsection{Multi-objective Dynamic Network Control}
\label{sec:ourproblem} Note that channel estimation is inherently
error-prone.  The degree of uncertainty in the estimate of the
current channel state depends on the previous channel observations,
and the dynamics of the channel. In this context, we define
\textit{information of an unexplored channel} as the uncertainty in
the channel state given its past observations.  This information can
be exactly quantified by using the entropy definition given by
Shannon. Accordingly, the scalar quantity $I_n(t)$ denotes the {\em
information} of channel state of user $n$ at the beginning of time
slot $t$ given past observations of the channel. For instance, the
information about a channel whose state was observed recently and
many times before is less than the channel which has not been probed
for a long time, since the uncertainty in the state of the latter is
higher.

Hence, we have two objectives.  First one is to schedule users so
that stability of the network is preserved.  The second closely
related objective is to probe users to acquire as much information
about their current channel state as possible.
\begin{itemize}
\item objective 1: $\max \sum_{n=1}^N
q_n(t)\hat{d}_{n}(t)$\\
\item objective 2: $\max \sum_{n=1}^N I_n(t)$
\end{itemize}
We seek a joint feasible policy $\hat{\pi}$ which determines a
subset of users probed by considering both objectives, and schedules
a user out of this subset according to Max-Weight algorithm. The
most common approach to find the solution of multi-objective
optimization problems is the weighted sum method. The problem under
\textit{dynamic} feedback model is given with a constraint which
ensures at most $M$ channels are probed at a given time slot:
\\
\textbf{Problem :}
\begin{align}
\max_{\hat{\pi} \in \Fc} \ &\sum_{n=1}^N \alpha_1
q_n(t)\hat{r}_{n}(\hat{\pi} ,t) +\alpha_2 I_n(t)  \label{eq:problem_m2}\\
&\text{s.t.}\ m \leq M,\notag
\end{align}
Note that the scheduling and probing decision depends not only on
the queue sizes and the estimated channel rates as in the original
Max-Weight algorithm, but also on the uncertainty in each channel
state given its past observations. Also, \eqref{eq:problem_m2}
exhibits the well-known ``Exploration vs. Exploitation" trade-off,
since the first term in the summation aims to stabilize the network
while the second term aims to maximize the information collected
about the channel states. In the following sections, we deal with a
modified version of this problem, where we divide the objective
function in \eqref{eq:problem_m2} by $\alpha_1$, and define a single
weight $\xi=\frac{\alpha_2}{\alpha_1}$. Note that when $\xi$ is
tuned to higher (lower) values, the channels are probed more (less)
frequently.

\section{Estimation of CSI with GPR}
\label{sec:gpr} The  problems given in \eqref{eq:problem_m2}
involves estimating $\hat{d}_{n}(t)$ from a set of past channel
observations. The problem of predicting or forecasting the value of
a variable from observations of other dependent variables is called
regression. There is a plethora of work for carrying out regression
analysis.  In this work, we employ Gaussian Process Regression (GPR)
as the technique for channel estimation. Before explaining how
channel state is estimated with GPR in detail, we first give the
main reasons behind this choice.
\begin{itemize}
\item  One of the well known methods is autoregressive (AR)
model-based techniques or linear regression. AR is parametric, in
that the channel function is defined in terms of a finite number of
unknown parameters. However, determining these parameters is a
difficult task especially when collecting data is costly and the
function varies over time. GPR is a nonparametric regression method
model. Thus, it can offer a more flexible framework for unknown
nonlinearities.

\item In contrast to other regression models, GPR provides a
simple way to measure the uncertainty in the estimation for any
given any set of CSI observations. AR model is lack of providing an
analytical way to measure the uncertainty of the estimation which is
important for our scheduling algorithm and we will mention next.

\item  The most attractive reason is that GPR can give
decisions with only using the most recent channel observations. This
is especially important for non-stationary channels, since previous
channel observations may become outdated and may not give much
information about current condition.
\end{itemize}

Let $\mathcal{D}_n(t)=(\textbf{c}_n,\boldsymbol \tau_n)$ denote the
set of observations for channel $n$ at the beginning of time slot
$t$, where $\textbf{c}_n=\{c_n^1, c_n^2, \dots, c_n^w\}$ denotes the
set of latest  $w$ CSI values taken at times, $\boldsymbol
\tau_n=\{\tau_n^1,\tau_n^2,\dots,\tau_n^w\}$, and $\tau_n^i < t$,
$\forall \tau_n^i \in \boldsymbol \tau_n$, $i\in \{1,2,\dots,w\}$.
We use GPR to predict the value of CSI, i.e., $\hat{c}_n(t)$ at the
beginning of time slot $t$, given $\mathcal{D}_n(t)$.

Let $p(c_n(t)| t,\mathcal{D}_n(t))$ be a posterior distribution of
channel $n$. According to GPR, a posterior distribution is Gaussian
with mean $\hat{c}_n(t)$ and variance $v_n(t)$. Specifically,
Gaussian process is specified by the kernel function, $k_n(\tau_n^i,
\tau_n^j)$, that describes the correlation of channel $n$ between
two of its measurements taken at times $\tau_n^i$ and $\tau_n^j$. It
is possible to choose any positive definite kernel function.
However, the most widely used is the squared exponential, i.e.,
Gaussian, kernel:
\begin{align}
k_n(\tau_n^i,
\tau_n^j)=\exp\left[-\frac{1}{2}(\tau_n^i-\tau_n^j)^2\right].\label{eq:kernel1}
\end{align}
Given $\mathcal{D}_n(t)$,  $\hat{c}_n(t)$ and variance $v_n(t)$ are
determined as follows:
\begin{align}
\hat{c}_n(t)&=\textbf{k}_n^T(t)\textbf{K}_n^{-1}\textbf{c}_n,\label{eq:mean1}\\
v_n(t)&=k_n(t,t)-\textbf{k}_n^T(t)\textbf{K}_n^{-1}\textbf{k}_n(t),
\label{eq:var1}
\end{align}
where $\textbf{K}_n$ is a $w \times w$ matrix composed of elements
$k_n(\tau_n^i, \tau_n^j)$ for $1\leq i,j\leq w$ and
$\textbf{k}_n(t)$ is a vector with elements $k(\tau_n^i,t)$ for
$\forall \tau_n^i \in \boldsymbol \tau_n$. Hence, the network
scheduler can easily predict the CSI of users at time $t$ by using
\eqref{eq:mean1}. Furthermore, the variance $v_n(t)$ is used to
measure the level of uncertainty in the estimations, i.e., $I_n(t)$
as discussed next.

Recall that the \textbf{entropy} of a random variable $A$ is defined
as $H(A)=\sum_s p_s \log_s(\frac{1}{p_s})$, where  $p(.)$ is the
probability distribution function of $A$. In our context, the
current realization of CSI, i.e., $c_n(t)$, is a random variable.
Accordingly, let $H_n^0(c_n(t)| t,\mathcal{D}_n(t))$ and
$H_n^1(c_n(t)| t,\mathcal{D}_n(t))$ denote the entropy of the random
variable $c_n(t)$ before and after the probing, respectively when
$\mathcal{D}_n(t)$ is given. If channel $n$ is probed at time $t$,
then $H_n^1(c_n(t)| t,\mathcal{D}_n(t))$ will be zero since the
channel state is known exactly. Otherwise, the uncertainty
increases, i.e., $H_n^1(c_n(t)| t,\mathcal{D}_n(t))
>H_n^0(c_n(t)| t,\mathcal{D}_n(t))$. Hence, the information acquired by probing
channel $n$ is the reduction in its uncertainty, which is simply the
difference between its entropies before and after the probing:
\begin{align*}
I_n(t)=H_n^0(c_n(t)| t,\mathcal{D}_n(t)) -H_n^1(c_n(t)|
t,\mathcal{D}_n(t)).
\end{align*}
The following Proposition is similar to the one given
in~\cite{Alpcan:Valuetools11}, and establishes that information
obtained by probing a channel is equal to the variance of the
estimate of the state of that channel.
\begin{proposition}
\label{prop:valuetools} Given $\mathcal{D}_n(t), \forall n=1,\ldots,
N$, finding the channel that has the highest information at time
slot $t$ is equal to finding the channel which has the highest
variance at that time slot, i.e.,
\begin{align}
i^*=\argmax_{n \in \mathcal{N}} I_n(t)=\argmax_{n \in \mathcal{N}}
v_n(t). \label{eq:proposition}
\end{align}
\end{proposition}
\begin{proof}
Since $H_n^1(c_n(t)| t,\mathcal{D}_n(t))=0$ after  probing, $I_n(t)$
is simply
\begin{align}
I_n(t)=H_n^0(c_n(t)| t,\mathcal{D}_n(t))\label{eq:infomax},
\end{align}
Note that according to GPR a posterior distribution of state of
channel given $\mathcal{D}_n$ is
\begin{align}
p(c_n(t)| t,\mathcal{D}_n)\sim \mathcal{N}(\hat{c}_n(t);v_n(t)).
\end{align}
Then, the entropy of  this Gaussian distribution is given by,
\begin{align}
H_n^0(c_n(t)| t,\mathcal{D}_n)=\frac{1}{2}\log(2\pi e v_n(t)).
\end{align}
Hence,
\begin{align*}
i^*=\argmax_{n \in \mathcal{N}} I_n(t)=\argmax_{n \in \mathcal{N}}
v_n(t).
\end{align*}
\end{proof}

\section{Joint Scheduling and Probing Algorithms} \label{sec:algs}
Here, we define an algorithm  for solving problem
\eqref{eq:problem_m2} when $\hat{c}_n(t)$ and $v_n(t)$ are
calculated as described in the previous section.

\subsection{Joint Algorithm Under \textit{dynamic} feedback model}
For given, $M$, $\xi$, $\boldsymbol q(t)$, $\beta$, and
$\hat{c}_{n}(t)$ and $v_n(t)$ determined by GPR for each user at
every time slot $t$, Algorithm  gives
$\hat{\pi}_1^*=(\hat{n}^*,\hat{m}^*, \hat{S}^*_{\hat{m}^*} )$:

\begin{algorithmic}
\textbf{ Algorithm }:\\
\textit{(1) probing decision}:\\
For each value of $m$, $m=\{1,2,\dots,M\}$, the scheduler calculates
the following weights for ever user $n$,
\begin{align*}
J_m^n \triangleq q_n(t)\hat{d}_{n}(t) + \xi I_n(t).\notag
\end{align*}
Then, the scheduler sorts $J_m^n$ in a descending order and sums the
first $m$ weights. The maximum of the sums is the maximum of
\eqref{eq:problem_m2}. Then, the corresponding $m$ and the first $m$
users in the order gives $\hat{m}^*$ and
$\hat{S}^*_{\hat{m}^*}$, respectively.\\
\textit{(2) scheduling decision}:\\
The base station acquires CSI of each user in
$\hat{S}^*_{\hat{m}^*}$ and user $n^* \in \hat{S}^*_{\hat{m}^*}$ is
scheduled to transmit,
\begin{align*}
\hat{n}^*=\argmax_{n \in \hat{S}^*_{\hat{m}^*}} q_n(t)d_{n}(t).
\end{align*}
\end{algorithmic}
\begin{proposition}
\label{prop:prop_alg2} Algorithm  solves \eqref{eq:problem_m2}.
\end{proposition}
\begin{proof}
The proof is straightforward and it is omitted here due to lack of
space.
\end{proof}

\section{Numerical Analysis}
\label{sec:sim} In our simulations, we model a single cell CDMA
downlink transmission utilizing high data rate (HDR). The base
station serves keeps a separate queue for each user. Time is slotted
with length $T_s=1.67$ ms as defined in HDR specifications. Packets
arrive at each slot according to Bernoulli distribution. For all
simulations, the wireless channel is modeled as \textit{correlated}
Rayleigh fading according to Jakes's model. Each user has different
Doppler frequency roughly characterizing how fast its channel
changes. The sampling rate for the simulations is 600Hz which also
corresponds to the slot size of HDR. Finally, the channel process is
non-stationary, i.e, the mean of the channel gain changes over time.
The transmission rate and the number of bits of a user transmits is
given as,
\begin{align}
\mu_n(t)&=\textrm{BW}\log_2\left(1 + \textrm{SNR}\times c_n(t)\right),\notag\\
R_n(t)&=T_s\times\mu_n(t)\notag,
\end{align}
where $\textrm{BW}$ is the channel bandwidth  set to
$\textrm{BW}=1.25$ MHz. The base station has power control to set
Signal-to-Noise-Ratio $\textrm{SNR}=10$ dB.

\begin{figure}
     \centering
     \includegraphics[width=0.9\columnwidth]{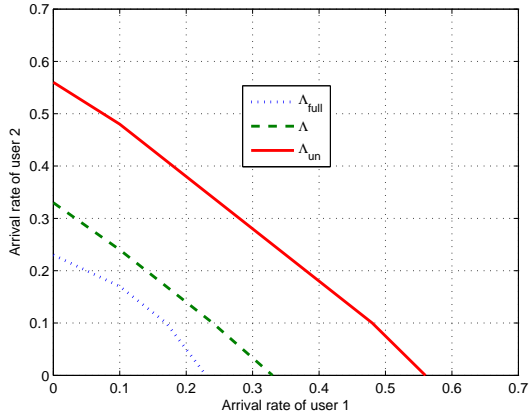}
     \caption{Rate regions under dynamic feedback model. }
     \label{fig:regions_sim}
\end{figure}
In the first simulation, we demonstrate the rate region of both
algorithms. There are only two users and the probing cost is
$\beta=0.3$. The arrival process for each user is again assumed to
be Bernoulli with a packet size of 631 Bytes. As depicted in
Figure~\ref{fig:regions_sim}, $\Lambda_{un}$ represents the
hypothetical rate region which is obtained when Max-Weight with full
CSI is used for scheduling. The boundary of this region cannot be
achieved in practice since $\beta$ is never zero. On the other
extreme, when all channels are probed at every slot without
neglecting the cost of probing, we obtain a rate region given as
$\Lambda_{full}$.  Meanwhile, Algorithm achieves the rate region
$\Lambda$. Clearly, by predicting the channel states by employing
Algorithm , we can increase the achievable rate region beyond
$\Lambda_{full}$. This is because by reducing the number of channels
probed at every slot, we can use a larger portion of time slot for
transmission of data.

\begin{figure}
     \centering
     \includegraphics[width=0.9\columnwidth]{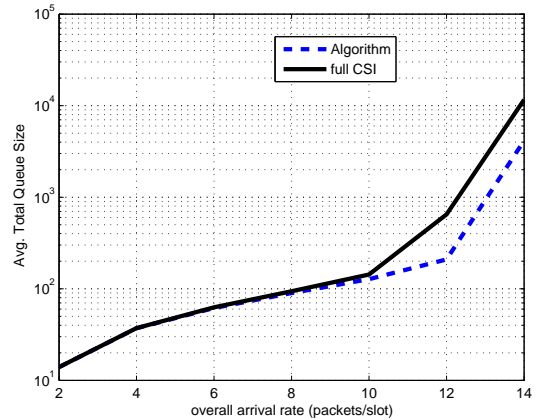}
     \caption{Average total queue sizes vs. overall mean arrival rate. }
     \label{fig:queuesize}
\end{figure}
Next, we show the performance of Algorithm when there are 20 users
in the network. The size of a packet is set to 128 bytes which
corresponds to the size of an HDR packet. Figure~\ref{fig:queuesize}
depicts the sum of the queue lengths vs. the overall arrival rate
when $\beta=0.02$. Clearly, as the overall arrival rate exceeds 10
packets/slot queue sizes suddenly increase within full CSI case and
the network becomes unstable. However, Algorithm improves over Max
Weight with full CSI by supporting the overall arrival rate of up to
12 packets/slot. Therefore,the proposed algorithm can achieve larger
rate region.

\section{Conclusion}
\label{sec:conclusion}  We have developed joint scheduling channel
probing algorithms for time-correlated and stationary/non-stationary
wireless channels. The proposed algorithm has been designed for the
channel probing model where the acquiring CSI of a use requires
$\beta$ fraction of the time slot. The proposed algorithm first
decides the set of channels that must be probed at the beginning of
each time slot. The set of channels is determined by considering not
only the queue sizes and estimated transmission rate but also the
information on each channel. We apply Gaussian Process technique to
predict CSI at each time slot based on the previous actual CSI
observed. In simulation results, we show that by applying GPR with
the proposed algorithm, the network can carry higher user traffic.


\bibliographystyle{IEEEtran}
\bibliography{IEEEabrv,ref}


\end{document}